\documentclass[onecolumn,12pt]{IEEEtran} 
\usepackage{amssymb}
\usepackage{amsthm,color}
\usepackage{latexsym}
\usepackage{graphicx}
\usepackage{fancyhdr}
\usepackage{bbm}
\usepackage{amsmath}
\usepackage{latexsym,amssymb,amsxtra,mathrsfs,times}

\renewcommand{\vec}[1]{\underline{#1}}

\newtheorem{thm}{Theorem}

\newtheorem{lemma}[thm]{Lemma}
\newtheorem{cor}[thm]{Corollary}

\theoremstyle{definition}

\newtheorem{rem}[thm]{Remark}
\newtheorem{exam}[thm]{Example}

\newcommand{\lf}[2]{\left(\frac{#1}{#2}\right)}
\newcommand{\tr}{{\rm Tr}}
\newcommand{\gf}{{\mathbb F}}

\begin{document}




\title{Weight Distribution of a Class of Cyclic Codes with Arbitrary Number of Zeros}

\author{Jing Yang,\thanks{J. Yang is with the Department of Mathematical Sciences, Tsinghua University,
Beijing, 100084, China (email: jingyang@math.tsinghua.edu.cn).}
Maosheng Xiong \thanks{M. Xiong is with the Department of Mathematics, The Hong Kong University of Science and Technology, Clear Water Bay, Kowloon, Hong Kong (email: mamsxiong@ust.hk).}
and
Cunsheng Ding\thanks{C. Ding is with the Department of Computer Science and Engineering, The Hong Kong University of Science and Technology, Clear Water Bay, Kowloon, Hong Kong, China (email: cding@ust.hk).}
}

\date{}


\maketitle


\renewcommand{\thefootnote}{}





\begin{abstract}
Cyclic codes have been widely used in digital communication systems and consume electronics as they have efficient encoding and decoding algorithms. The weight distribution of cyclic codes has been an important topic of study for many years. It is in general hard to determine the weight distribution of linear codes. In this paper, a class of cyclic codes with any number of zeros are described and their weight distributions are determined.
\end{abstract}

\begin{keywords}
Cyclic codes, Gaussian periods, linear codes, weight distribution.
\end{keywords}

\section{Introduction}\label{sec-into}

Throughout this paper, let $p$ be a prime, $q=p^s$, $r=q^m$ for some integers $s,m\geqslant 1$. Let $\mathbb{F}_r$ be a finite field of order $r$ and $\gamma$ be a generator of the multiplicative group $\mathbb{F}_{r}^*:=\mathbb{F}_r \setminus \{0\}$. An $[n,\kappa,d]$-linear code $\mathcal{C}$ over $\mathbb{F}_q$ is a $\kappa$-dimensional subspace of $\mathbb{F}_{q}^n$ with minimum (Hamming) distance $d$. It is called cyclic if any $(c_0,c_1,\cdots ,c_{n-1})\in \mathcal{C}$ implies $(c_{n-1},c_0,\cdots,c_{n-2})\in \mathcal{C}$.

Consider the one-to-one linear map defined by
$$\begin{array}{cccl}
\sigma:& \mathcal{C}&\rightarrow &R=\mathbb{F}_{q}[x]/(x^n-1)\\
 &(c_0,c_1,\cdots ,c_{n-1})&\mapsto&c_0+c_1x+\cdots +c_{n-1}x^{n-1}.
\end{array}$$
Then $\mathcal{C}$ is a cyclic code if and only if $\sigma(\mathcal{C})$ is an ideal of the ring $R$.
Since $R$ is a principal ideal ring, there exists a unique monic polynomial $g(x)$ with least degree satisfying $\sigma(\mathcal{C})=g(x)R$ and $g(x)\mid (x^n-1)$. Then $g(x)$ is called the \textit{generator polynomial} of $\mathcal{C}$ and $h(x)=(x^n-1)/g(x)$ is called the \textit{parity-check polynomial} of $\mathcal{C}$. If $h(x)$ has $t$ irreducible factors over $\mathbb{F}_{q}$, we say for simplicity such a cyclic code $\mathcal{C}$ to \textit{have $t$ zeros}. (In the literature some authors call $\mathcal{C}$ ``the dual of a cyclic code with $t$ zeros''.)

Denote by $A_i$ the number of codewords with Hamming weight $i$ in $\mathcal{C}$. The {\em weight enumerator} of $\mathcal{C}$ with length $n$ is defined by
$$1+A_1z+A_2z^2+ \cdots + A_nz^n.$$
The sequence $(A_0,A_1,\cdots ,A_n)$ is called the \textit{weight distribution} of $\mathcal{C}$. The study of the weight distribution of a linear code is important in both theory and application due to the following:
\begin{itemize}
  \item The weight distribution of a code gives the minimum distance and thus the error correcting capability of the code.
  \item The weight distribution of a code allows the computation of the error probability of error detection and correction
            with respect to some algorithms \cite{Klov}.
\end{itemize}
The problem of determining the weight distribution of linear codes is in general very difficult and remains open
for most linear codes. For only a few special classes the weight distribution is known. For example, the weight
distribution of some irreducible cyclic codes is known (\cite{AL06,BM72,BM73,McE74,D-Y12,Rao10}).  For cyclic codes with two zeros the weight distribution is known in some special cases (\cite{Ding2,Ding1,F-M12,Tang12,Vega12,Xiong1,Xiong2}). The weight distribution is also known for some other linear and cyclic codes (\cite{B-M10,FL08,Feng12,LF08,LTW,McG,M-R07,Mois07,Schoof,Y-C-D06,Zeng10}).

The objectives of this paper are to describe a new class of cyclic codes with arbitrary number of zeros and to determine their
weight distributions. This paper is organized as follows. Section \ref{sec-codess} defines this class of cyclic codes.
Section \ref{sec-pre} introduces some mathematical tools such as group characters, cyclotomy and Gaussian periods that will be needed later in this
paper.  Section \ref{sec-main} deals with the weight distribution of the class of cyclic codes under special conditions.
Section \ref{sec-conclusion} concludes this paper.

\section{The class of cyclic codes}\label{sec-codess}

From now on, we make the following  assumptions for the rest of this paper.

\vspace{.1cm}
\noindent
{\bf The Main Assumptions:}
\emph{Let $r=q^m=p^{sm}$ be a prime power for some positive integers $s,m$ and let $e \geqslant t \geqslant 2$. Assume that}
\begin{itemize}
\item[ i)]  \emph{$a \not \equiv 0 \pmod{r-1} \mbox{ and } e|(r-1)$;}

\item[ ii)] \emph{$a_i \equiv a+\frac{r-1}{e}\Delta_i \pmod{r-1},\, 1\leqslant i \leqslant t$, where $\Delta_i \not \equiv \Delta_j \pmod{e}$ for any $ i \ne j$ and \\ $\gcd(\Delta_2-\Delta_1,\ldots,\Delta_t-\Delta_1,e)=1$;}

\item[ iii)] \emph{$
\deg h_{a_1}(x)=\cdots=\deg h_{a_t}(x)=m, \mbox{ and } h_{a_i}(x) \neq h_{a_j}(x)$ for any $1\leqslant i\neq j\leqslant t$, where $h_{a}(x)$ is the minimal polynomial of $\gamma^{-a}$ over $\mathbb{F}_q$.} $\quad\hfill \blacksquare$


\end{itemize}

We remark that Condition iii) can be met by a simple criterion stated in Lemma \ref{lem-hi}. From what follows, define
\[\delta=\gcd(r-1,a_1,a_2,\cdots ,a_{t}), \quad n=\frac{r-1}{\delta} \]
and
\[ N=\gcd \left(\frac{r-1}{q-1},a e\right).\]
It is easy to verify that
\[e \delta \mid N(q-1). \]
The class of cyclic codes considered in this paper is defined by
\begin{equation}\label{def}
\mathcal{C}=\left\{ c(x_1,x_2,\cdots,x_{t})=\left(Tr_{r/q}\left(\sum_{j=1}^t x_j \gamma^{a_ji}  \right)\right)_{i=0}^{n-1}~:~x_1,\cdots,x_{t}\in\mathbb{F}_{r} \right\},\end{equation}
where $Tr_{r/q}$ denotes the trace map from $\mathbb{F}_{r}$ to $\mathbb{F}_{q}$.
It follows from Delsarte's Theorem \cite{Delsarte} that the code $\mathcal{C}$ is an $[n,tm]$ cyclic code over $\mathbb{F}_{q}$ with parity-check
polynomial $h(x)=h_{a_1}(x)\cdots h_{a_{t}}(x)$. This code $\mathcal{C}$ may contain many cyclic codes studied in the literature as special cases. In particular, when $t=2$, $a_0=\frac{q-1}{h},a_1=\frac{q-1}{h}+\frac{r-1}{e}$ for positive integers $e,h$ such that $e|h$ and $h|(q-1)$, the code $\mathcal{C}$ has been studied in \cite{Ding1,Ding2,Tang12,Xiong1,Xiong2,F-M12}.

In the definition of $\mathcal{C}$ we choose integers $a_1,a_2,\cdots ,a_{t}$ from a set of arithmetic
sequence with common difference $\frac{r-1}{e}$ modulo $r-1$. This choice of these $a_i$'s allows us to compute the weight distribution
of the code $\mathcal{C}$. If the integers $a_i$ are not chosen in this way, it might be difficult to find the weight
distribution. The conditions in the Main Assumptions are to guarantee that the dimension of $\mathcal{C}$ is equal to $mt$.

\section{Group characters, cyclotomy and Gaussian periods}\label{sec-pre}

Let $\tr_{r/p}$ denote the trace function from $\mathbb{F}_{r}$ to $\mathbb{F}_{p}$.
An {\em additive character} of $\mathbb{F}_{r}$ is a nonzero function $\psi$
from $\mathbb{F}_{r}$ to the set of complex numbers such that
$\psi(x+y)=\psi(x) \psi(y)$ for any pair $(x, y) \in \mathbb{F}_{r}^2$.
For each $b\in \mathbb{F}_{r}$, the function
\begin{eqnarray}\label{dfn-add}
 \psi_b(c)=e^{2\pi \sqrt{-1} \tr_{r/p}(bc)/p} \ \ \mbox{ for all }
c\in\mathbb{F}_{r}
\end{eqnarray}
defines an additive character of $\mathbb{F}_{r}$. When $b=0$,
$\psi_0(c)=1 \mbox{ for all } c\in\mathbb{F}_{r},$
and is called the {\em trivial additive character} of
$\mathbb{F}_{r}$. When $b=1$, the character $\psi_1$ in (\ref{dfn-add}) is called the
{\em canonical additive character} of $\mathbb{F}_{r}$. For any $x\in \mathbb{F}_{r}$, one can easily check the following orthogonal property of additive characters, which we need in the sequel,
\begin{equation}\label{add-orth}
    \frac{1}{r}\sum\limits_{x\in \mathbb{F}_{r}}\psi(ax)=\left\{
      \begin{array}{ll}
        1, & \hbox{if $a=0$;} \\
        0, & \hbox{if $a\in \mathbb{F}^*_{r}$.}
      \end{array}
    \right.
\end{equation}

Let $r-1=l L$ for two positive integers $l \geqslant 1$ and $L\geqslant 1$, and let
$\gamma$ be a fixed primitive element of $\mathbb{F}_{r}$.
Define $C_{i}^{(L,r)}=\gamma^i \langle \gamma^{L} \rangle$ for $i=0,1,...,L-1$, where
$\langle \gamma^{L} \rangle$ denotes the
subgroup of $\mathbb{F}_{r}^*$ generated by $\gamma^{L}$. The cosets $C_{i}^{(L,r)}$ are
called the {\em cyclotomic classes} of order $L$ in $\mathbb{F}_{r}$.
The {\em cyclotomic numbers} of order $L$ are
defined by
\begin{eqnarray*}
(i, j)^{(L,r)}=\left|(C_{i}^{(L,r)}+1) \cap C_{j}^{(L,r)}\right|
\end{eqnarray*}
for all $0 \leqslant i,j \leqslant  L-1$.

Cyclotomic numbers of order 2 are given in the following lemma \cite{B-E-W} and will be needed in the sequel.

\begin{lemma}\label{lem-cycNo-N=2}
The cyclotomic numbers of order 2 are given by
\begin{itemize}
  \item $(0,0)^{(2,r)}=\frac{(r-5)}{4};~(0,1)^{(2,r)}=(1,0)^{(2,r)}=(1,1)^{(2,r)}=\frac{(r-1)}{4}$ if $r\equiv 1\pmod{4}$; and
  \item $(0,0)^{(2,r)}=(1,0)^{(2,r)}=(1,1)^{(2,r)}=\frac{(r-3)}{4};~(0,1)^{(2,r)}=\frac{(r+1)}{4}$ if $r\equiv 3\pmod{4}$.
\end{itemize}
\end{lemma}

\vspace{.2cm}
The {\em Gaussian periods} of order $L$ are defined by
$$
\eta_i^{(L,r)} =\sum_{x \in C_i^{(L,r)}} \psi(x), \quad i=0,1,..., L-1,
$$
where $\psi$ is the canonical additive character of $\mathbb{F}_{r}$.

The values of the Gaussian periods are in general very hard to compute.
However, they can be computed in a few cases. We will need the following lemmas whose proofs can be found in \cite{B-E-W} and \cite{Myer}.

\vspace{.2cm}
\begin{lemma}\label{lem-degree2}
When $L=2$, the Gaussian periods are given by
\begin{eqnarray*}
\eta_0^{(2,r)}=
\left\{
\begin{array}{ll}
\frac{-1+(-1)^{s\cdot m-1}r^{1/2}}{2}, & \mbox{if $p\equiv 1 \pmod{4}$} \\
\frac{-1+(-1)^{s\cdot m-1}(\sqrt{-1})^{s\cdot m} r^{1/2}}{2}, & \mbox{if $p\equiv 3 \pmod{4}$}
\end{array}
\right.
\end{eqnarray*}
and
$
\eta_1^{(2,r)} = -1 - \eta_0^{(2,r)}.
$
\end{lemma}

\vspace{.2cm}
\begin{lemma}\label{lem-period2}
Let $L=3$. If $p \equiv 1 \pmod{3}$, and $s m \equiv 0 \pmod{3}$, then
\begin{equation*}\left\{
                 \begin{array}{l}
                   \eta_0^{(3,r)}=\frac{-1-c_1r^{1/3}}{3} \\
                   \eta_1^{(3,r)}=\frac{-1+\frac{1}{2}(c_1+9d_1)r^{1/3}}{3} \\
                   \eta_2^{(3,r)}=\frac{-1+\frac{1}{2}(c_1-9d_1)r^{1/3}}{3}.
                 \end{array}
               \right.\end{equation*}
where $c_1$ and $d_1$ are given by $4p^{s\cdot m/3}=c_1^2+27d_1^2$, $c_1 \equiv 1 \pmod{3}$ and
$\gcd(c_1,p)=1$.
\end{lemma}

In a special case, the so-called \textit{semiprimitive case}, the Gaussian periods are known and are described in the
following lemma \cite{BM72,Myer}.

\vspace{.2cm}
\begin{lemma}\label{lem-semip}
Assume that $L>2$ and there exists a positive integer $j$ such that $p^j \equiv -1 \pmod{L}$, and
the $j$ is the least such. Let $r=p^{2jv}$ for some integer $v$.

(a) If $v$, $p$ and $(p^j+1)/L$ are all odd, then
\begin{eqnarray*}
\begin{array}{l}
\eta_{L/2}^{(L,r)}=\frac{(L-1)\sqrt{r}-1}{L}, \ \
\eta_{k}^{(L,r)}=-\frac{\sqrt{r}+1}{L}  \mbox{ for } k \ne L/2.
\end{array}
\end{eqnarray*}

(b) In all other cases,
\begin{eqnarray*}
\begin{array}{l}
\eta_{0}^{(L,r)}=\frac{(-1)^{v+1}(L-1)\sqrt{r}-1}{L}, \ \
\eta_{k}^{(L,r)}=\frac{(-1)^v \sqrt{r}-1}{L} \mbox{ for } k \ne 0.
\end{array}
\end{eqnarray*}
\end{lemma}

In another special case,  the so-called \textit{quadratic residue (or index 2) case}, the Gaussian periods can be also computed.
The results below are from \cite{BM73} or \cite{D-Y12}.

\begin{lemma}\label{Index2}
Let $3\neq L\equiv 3\pmod{4}$ be a prime, $p$ be a quadratic residue modulo $L$
and $\frac{L-1}{2}\cdot k=sm$ for some positive integer $k$. Let $h_L$ be the ideal class number of $\mathbb{Q}(\sqrt{-L})$ and $a,b$ be integers satisfying
\begin{equation}\label{ab-condition}
    \left\{
      \begin{array}{l}
        a^2+Lb^2=4p^{h_L}\\
        a\equiv -2p^{\frac{L-1+2h_L}{4}}\pmod{L}\\
        b>0,p\nmid b.
      \end{array}
    \right.
\end{equation}
Then, the Gaussian periods of order $L$ are given by
\begin{equation}\label{index2-period}
    \left\{
       \begin{array}{ll}
         \eta_0^{(L,r)}=\frac{1}{L}(P^{(k)}A^{(k)}(L-1)-1) &  \\
         \eta_u^{(L,r)}=\eta_1=\frac{-1}{L}(P^{(k)}A^{(k)}+P^{(k)}B^{(k)}L+1), & \hbox{if $\lf{u}{L}=1$} \\
         \eta_u^{(L,r)}=\eta_{-1}=\frac{-1}{L}(P^{(k)}A^{(k)}-P^{(k)}B^{(k)}L+1), & \hbox{if $\lf{u}{L}=-1$,}
       \end{array}
     \right.
\end{equation}
where
\begin{equation}\label{PAB}
    \left\{
      \begin{array}{l}
      P^{(k)}=(-1)^{k-1}p^{\frac{k}{4}(L-1-2h_L)} \\
      A^{(k)}=\mathrm{Re}(\frac{a+b\sqrt{-L}}{2})^k\\
      B^{(k)}=\mathrm{Im}(\frac{a+b\sqrt{-L}}{2})^k\big/\sqrt{L}.
      \end{array}
    \right.
\end{equation}
\end{lemma}

\section{The weight distributions of this class of codes under certain conditions}\label{sec-main}

We first provide the following criterion that guarantees Condition iii) in the Main Assumptions.

\begin{lemma}\label{lem-hi}
(a) Suppose that for any proper factor $\ell$ of $m$ (i.e. $\ell \mid m$ and $\ell<m$) we have
$$\frac{r-1}{q^\ell - 1}\nmid N.$$
Then Condition iii) in the Main Assumptions holds.
\par (b) In particular, if $N \leqslant \sqrt{r}$, then Condition iii) in the Main Assumptions is met.
\end{lemma}

\begin{proof}
Suppose Condition iii) does not hold, then there exists a positive integer $h <m$ such that
$$a_iq^h \equiv a_j \pmod{r-1}$$
for some $1 \leqslant i, j\leqslant t$. Reducing modulo $(r-1)/e$ we obtain that
\begin{equation}\label{equ-qk}
aq^h \equiv a\pmod{\frac{r-1}{e}}.
\end{equation}
Hence $(r-1)\mid ae(q^h-1)$. Since $\gcd(r-1,q^h-1)=q^\ell-1$ where $\ell=\gcd(h,m)$, it then follows from
(\ref{equ-qk}) that
\begin{equation}\label{equ-qk2}
\left.\frac{r-1}{q^\ell-1}~\right| ~ae.\end{equation}
Hence
$$\left.\frac{r-1}{q^h-1}=\gcd \left( \frac{r-1}{q-1},\frac{r-1}{q^\ell-1} \right)~\right| ~\gcd \left( \frac{r-1}{q-1},ae\right)=N.$$
Since $\ell|m$ and $\ell<m$, this contradicts the condition of the lemma. Thus Part (a) is proved.

Part (b) of Lemma \ref{lem-hi} can be derived from Part (a) directly. For any proper factor $\ell$ of $m$, we have
$\ell \leqslant m/2$. Thus $\frac{r-1}{q^\ell-1}$ can not be a divisor of $N$ which is at most $\sqrt{r}$ because
$$\frac{r-1}{q^\ell -1}\geqslant \frac{r-1}{\sqrt{r}-1}=\sqrt{r}+1.$$
This completes the proof of Lemma \ref{lem-hi}.
\end{proof}

We now consider the weight distribution of the cyclic code $\mathcal{C}$ given in (\ref{def}). In order to find the Hamming weight of the codeword $c(x_1,\cdots,x_{t})$, it suffices to consider a new codeword $c'(x_1,\cdots,x_{t})$ given by
\[c'(x_1,\ldots,x_t)=\left(Tr_{r/q}\left(\sum_{j=1}^t x_j \gamma^{a_ji}  \right)\right)_{i=0}^{r-2},\]
because clearly $c'(x_1,\cdots,x_{t})$ is the codeword $c(x_1,\cdots,x_{t})$ repeating itself $\delta$ times and hence
\[w_H(c(x_1,\cdots,x_{t}))=\frac{w_H(c'(x_1,\cdots,x_{t}))}{\delta}. \]

Let $\psi_q(x)=\exp(2\pi\sqrt{-1}Tr_{q/p}(x)/p)$ be the canonical additive character of $\mathbb{F}_{q}$. Then $\psi=\psi_q\circ Tr_{r/q}$ is the canonical additive character of $\mathbb{F}_{r}$. Using the orthogonal relation (\ref{add-orth}), we know that the Hamming weight of the codeword $c'(x_1,\cdots,x_{t})$ is given by
\begin{eqnarray*}
\lefteqn{w_H(c'(x_1,\cdots,x_{t})) } \\
&=&r-1-\sum\limits_{i=0}^{r-2}\frac{1}{q}\sum\limits_{y\in
\mathbb{F}_{q}}\psi_q[yTr_{r/q}(x_1\gamma^{a_1i}+\cdots+x_{t}\gamma^{a_{t}i})]\\
&=&r-1-\frac{r-1}{q}-\frac{1}{q}\sum\limits_{y\in\mathbb{F}_{q}^*}
\sum\limits_{i=0}^{r-2}\psi[y\gamma^{ai} (x_1\gamma^{(a_1-a)i}+\cdots+x_{t}\gamma^{(a_{t}-a)i})]\\
&=&\frac{(r-1)(q-1)}{q}-\frac{1}{q}\sum\limits_{y\in\mathbb{F}_{q}^*}
\sum\limits_{i=0}^{r-2}\psi[y\gamma^{ai} (x_1\gamma^{\frac{r-1}{e}\Delta_1i}+\cdots+x_{t}\gamma^{\frac{r-1}{e}\Delta_{t}i})]
\end{eqnarray*}
{From Condition i) of the Main Assumptions, we know that $e\mid (r-1)$, hence we can write $i=ej+h$ for $0 \leqslant j\leqslant \frac{r-1}{e}-1$ and $0\leqslant h\leqslant e-1$. Denote
\begin{equation}\label{equ-beta-g}\beta_\tau=\gamma^{\frac{r-1}{e}\Delta_\tau} \mbox{\ for\ } 1 \leqslant  \tau \leqslant t, \mbox{ and } g=\gamma^{a}.\end{equation}
Hence}
\begin{eqnarray*}
\lefteqn{w_H(c'(x_1,\cdots,x_{t}))} \\
&=&\frac{(r-1)(q-1)}{q}-\frac{1}{q}\sum\limits_{y\in\mathbb{F}_{q}^*}\sum\limits_{j=0}^{\frac{r-1}{e}-1}\sum\limits_{h=0}^{e-1} \psi[y \gamma^{aej}\gamma^{ah} (x_1\beta_1^{h}+\cdots+x_{t}\beta_{t}^{h})].\\
&=&\frac{(r-1)(q-1)}{ q}-\frac{1}{q}\sum\limits_{l=0}^{q-2}\sum\limits_{j=0}^{\frac{r-1}{e}-1}\sum\limits_{h=0}^{e-1} \psi[\gamma^{N\{\frac{r-1}{N(q-1)}l+\frac{ae}{N}j\}}g^{h} (x_1\beta_1^{h}+\cdots+x_{t}\beta_{t}^{h})],
\end{eqnarray*}
where we defined $N=\gcd(\frac{r-1}{q-1},ae)$ in Section \ref{sec-codess}. For each $X \pmod{\frac{r-1}{N}}$, we consider the number of solutions $(l,j)$ with $0\leqslant l\leqslant q-2,\ 0\leqslant j\leqslant \frac{r-1}{e }-1$ such that
\begin{equation}\label{equ-lj}
   \frac{r-1}{N(q-1)}l+\frac{ae}{N}j\equiv X \pmod{\frac{r-1}{N}}.
\end{equation}
Reducing modulo $\frac{r-1}{N(q-1)}$, we find that
$$\frac{ae}{N}j\equiv X \pmod{\frac{r-1}{N(q-1)}}.$$
This has a unique solution for $j$ modulo $\frac{r-1}{N(q-1)}$, hence the number of $j$ for $0\leqslant j\leqslant \frac{r-1}{e}-1$ that satisfies the equation is
$$\frac{(r-1)/e }{(r-1)/N(q-1)}=\frac{N(q-1)}{e }.$$
{For each such solution $j$, returning to Equation (\ref{equ-lj}), we find}
$$l\equiv \frac{x-\frac{ae}{N}j}{(r-1)/N(q-1)}\pmod{q-1},$$
this means there is a unique such $l$ with $0\leqslant l\leqslant q-2$.
Therefore
\begin{eqnarray*}
\lefteqn{w_H(c'(x_1,\cdots,x_{t}))} \\
&=&\frac{(r-1)(q-1)}{q}-\frac{N(q-1)}{ eq}\sum\limits_{h=0}^{e-1}\sum\limits_{X=0}^{\frac{r-1}{N}-1} \psi[\gamma^{NX}g^{h}(\sum\limits_{\tau=1}^{t} x_{\tau} \beta_{\tau}^{h})]\\
&=&\frac{(r-1)(q-1)}{q}-\frac{N(q-1)}{ eq}\sum\limits_{h=0}^{e-1}\sum\limits_{z\in C_0^{(N,r)}} \psi[zg^{h} (\sum\limits_{\tau=1}^{t} x_{\tau} \beta_\tau^{h})]\\
&=&\frac{(r-1)(q-1)}{q}-\frac{N(q-1)}{ eq}\sum\limits_{h=0}^{e-1}
\bar\eta^{(N,r)}_{g^{h}\cdot\sum\limits_{\tau=1}^{t} x_\tau \beta_\tau^{h}}.
\end{eqnarray*}
Here we write $\bar\eta^{(N,r)}_{v}=\sum\limits_{z\in C_{0}^{(N,r)}}\psi(vz)$ for any $v\in\mathbb{F}_{r}$ and call these $\bar\eta^{(N,r)}_v$ the \textit{modified Gaussian periods}, since
$$\left\{
    \begin{array}{l}
     \bar\eta^{(N,r)}_0=\frac{r-1}{N}\\
     \bar\eta^{(N,r)}_{\gamma^{i}}=\eta_i^{(N,r)}\quad \hbox{ for $0\leqslant i\leqslant N-1$,}
    \end{array}
  \right.$$
where these $\eta_i^{(N,r)}$ are the classical Gaussian periods. We conclude that
\begin{eqnarray*}
w_H(c(x_1,\cdots,x_{t})) =\frac{(r-1)(q-1)}{q \delta}-\frac{N(q-1)}{ eq \delta}\sum\limits_{h=0}^{e-1}
\bar\eta^{(N,r)}_{g^{h}\cdot\sum\limits_{\tau=1}^{t} x_\tau \beta_\tau^{h}}.
\end{eqnarray*}

Thus, to compute the weight distribution of cyclic code $\mathcal{C}$, it suffices to compute the value distribution of the sum
\begin{equation}\label{equ-Tx}
T(x_1,\cdots ,x_{t}):=\sum\limits_{h=0}^{e-1}\bar\eta^{(N,r)}_{g^{h}
\cdot\sum_{\tau=1}^{t} x_\tau \beta_\tau^{h}}.\end{equation}
This is in general a difficult problem. We will deal with it for some special cases in the next two subsections.

\subsection{The case of $t=e \geqslant 2$}\label{sec-e}

In this case the set $\{\Delta_i : 1 \leqslant i \leqslant e\}$ is a complete residue system modulo $e$, so we may take $\Delta_1=0,\Delta_2=1,\cdots,\Delta_{e}=e-1$. Define $\beta:=\beta_2$, then $\beta=\gamma^{(r-1)/e}$ is an $e$-th root of unity in $\mathbb{F}_{r}$ and $\beta_i=\beta^{i-1}$ for $1 \leqslant i \leqslant t$. We now present a key observation, which enables us to count the frequency of the weights in a simple and clear way. Consider the
linear transform $\varphi:~\mathbb{F}_{r}^e \rightarrow \mathbb{F}_{r}^e$ given by
\begin{equation}\label{phi-transform1}
 \varphi\left(
          \begin{array}{l}
            x_1 \\
            x_2\\
            \vdots  \\
            x_{e} \\
          \end{array}
        \right)
=\left(\begin{array}{llll}
           1&1&\cdots&1 \\
           1&\beta&\cdots&\beta^{e-1} \\
           1&\beta^2&\cdots&\beta^{2(e-1)} \\
           \vdots &\vdots &&\vdots \\
           1&\beta^{e-1}&\cdots&\beta^{(e-1)^2}
         \end{array}
       \right)
\left(\begin{array}{l}
            x_1 \\
            x_2\\
            \vdots  \\
            x_{e} \\
          \end{array}\right)=
\left(\begin{array}{l}
            y_0 \\
            y_1\\
            \vdots  \\
            y_{e-1} \\
          \end{array}\right).
\end{equation}
Since $1,\beta,\beta^2,\cdots \beta^{e-1}$ are distinct, the Vandermonde matrix
\begin{equation}\label{matrix}A:=\left(
         \begin{array}{llll}
           1&1&\cdots&1 \\
           1&\beta&\cdots&\beta^{e-1} \\
           1&\beta^2&\cdots&\beta^{2(e-1)} \\
           \vdots &\vdots &&\vdots\\
           1&\beta^{e-1}&\cdots&\beta^{(e-1)^2}
         \end{array}
       \right)\end{equation}
 is invertible. We then have the following observation.
\\[2mm]
\textbf{Observation A}: \textit{The map $\varphi$ is an isomorphism from $\mathbb{F}_{r}^{e}$ to $\mathbb{F}_{r}^{e}$.} Then $y_0,\cdots ,y_{e-1}$ independently run over $\mathbb{F}_{r}$ as $x_1,\cdots ,x_e$ run over $\mathbb{F}_{r}$.\hfill$\blacksquare$
\vskip3mm

Observation A means that it suffices to study the value distribution of
$$\tilde{T}(y_0,\cdots ,y_{e-1}):=\sum\limits_{h=0}^{e-1}\bar\eta^{(N,r)}_{g^hy_h},\quad \forall(y_0,\cdots ,y_{e-1})\in \mathbb{F}_{r}^e.$$

\vskip2mm
\subsubsection{The subcase of $t=e$ and $N=1$}

When $N=1$, we have $e \delta \mid (q-1)$, $C_0^{(1,r)}=\langle\gamma\rangle=\mathbb{F}_{r}^*$, and
\begin{eqnarray*}
\bar\eta^{(1,r)}_{v}=
\left\{\begin{array}{ll}
 r-1, & \hbox{if $v=0$ } \\
 -1, & \hbox{if $v\in\mathbb{F}_{r}^*$.}
  \end{array}\right.
\end{eqnarray*}
Hence the value $\tilde{T}(y_0,\cdots ,y_{e-1})$ depends only on the total number of $i$'s such that $y_i =0$. Denote this number by $u$ where $0\leqslant u\leqslant e$. Then
$$\tilde{T}(y_0,\cdots ,y_{e-1})=
  \begin{array}{l}
    u(r-1)+(e-u)(-1)=ur-e,
  \end{array}
$$
and the number of times that  $\tilde{T}$ takes this value for such $(y_0,\ldots,y_{e-1})$'s is clearly $\binom{e}{u}(r-1)^{e-u}$.
Thus, we have the result below.

\begin{thm}\label{thm-e1}
Under the Main Assumptions, when $N=1$ and $e=t \geqslant 2$, the set $\mathcal{C}$ defined by (\ref{def}) is an $e$-weight $[n,tm,\frac{(q-1)r}{\delta eq}]$ cyclic code. The weight distribution of $\mathcal{C}$ is listed in Table \ref{Table2}.
\end{thm}
\begin{table}[ht]
\caption{The weight distribution of $\mathcal{C}$ when $N=1$ and $e=t\geqslant 2$.}\label{Table2}
\begin{center}{
\begin{tabular}{|c|c|}
  \hline
  Weight & Frequency $\quad (0\leqslant u\leqslant e)$ \\\hline\hline
  $\frac{(q-1)r}{\delta eq}\cdot u$ & $\binom{e}{u}(r-1)^u$ times \\
  \hline
\end{tabular}}
\end{center}
\end{table}

\begin{exam}
Let $(q, m, e, t)=(3,3,2,2)$. Let $\gamma$ be the generator of $\gf_r^*$ with $\gamma^3 + 2\gamma + 1=0$.
Let $a=1$. Then $N=1$, $(a_1, a_2)=(1,14)$ and
$$
h_{a_1}(x)=x^3 + 2x^2 + 1, \ h_{a_2}(x)= x^3 + x^2 + 2.
$$
The parity-check polynomial of $\mathcal{C}$ is then
$
h(x)=x^6 + 2x^4 + 2x^2 + 2.
$
The code $\mathcal{C}$ is a $[26,6,9]$ ternary cyclic code with weight enumerator $1+52z^9+676z^{18}$.
\end{exam}

\subsubsection{The subcase of $t=e$ and $N\geqslant 2$}

In this case, we first give a general result stated in the following theorem.

\begin{thm}\label{thm-e2}
Suppose that the Gaussian periods $\eta_{i}^{(N,r)}$ of order $N$ have $\mu$ distinct values $\{\eta_1,\eta_2,\cdots ,\eta_{\mu}\}$, and each $\eta_i$ corresponds to $\tau_i$ cyclotomic classes for $1\leqslant i\leqslant \mu$. (Note that $\tau_1+\cdots +\tau_\mu=N$.) Then the cyclic code $\mathcal{C}$ defined in (\ref{def}) is an $[n,em]$ code over $\mathbb{F}_q$ with at most $\binom{\mu+e}{e}-1$ nonzero weights. Moreover, for any non-negative integers $u_0,u_1,\cdots,u_\mu$ such that $\sum_{j=0}^\mu u_j=e$, the weight distribution of $\mathcal{C}$ is listed in Table \ref{Table-e2}.
\end{thm}
\begin{table}[ht]
\caption{The weight distribution of $\mathcal{C}$ when $e=t,N\geqslant 2$.}\label{Table-e2}
\begin{center}{
\begin{tabular}{|c|c|}
  \hline
  Weight & Frequency $\quad (\sum_{j=0}^\mu u_j=e)$ \\\hline\hline
  $\frac{(q-1)}{\delta eq}\sum\limits_{j=1}^\mu u_j (r-1-N\eta_j)$ & $\frac{e!}{u_0!u_1!\cdots u_\mu!}\left(\frac{r-1}{N}\right)^{e-u_0}\prod_{j=1}^{\mu}\tau_j^{u_j}$ times \\\hline
\end{tabular}}
\end{center}
\end{table}

\begin{proof}
We just need to compute the value distribution of $\tilde{T}(y_0,y_1,\cdots ,y_{e-1})$. By Observation A, $y_0,gy_1,\cdots ,g^{e-1}y_{e-1}$ run over each $C_i^{(N,r)}\ (0\leqslant i\leqslant N-1)$ independently and uniformly. Suppose among the $g^iy_i$'s, exactly $u_0$ of them takes on $0$ and $u_i$ of them correspond to $\tau_i$ cyclotomic classes with value $\eta_i$ for $1 \leqslant i \leqslant \mu$ respectively. Then $\tilde{T}(y_0,y_1,\cdots ,y_{e-1})$ has at most $\binom{\mu+e}{e}$ possible values. More precisely, it takes on the value
\[ u_0\bar{\eta}_0+\sum\limits_{j=1}^\mu u_j \eta_j=u_0\frac{r-1}{N}+\sum\limits_{j=1}^\mu u_j \eta_j.\]
with the frequency of
$$\binom{e}{u_0}\binom{e-u_0}{u_1}\binom{e-u_0-u_1}{u_2}\cdots \binom{u_{\mu-1}+u_\mu}{u_{\mu-1}}\left(\frac{r-1}{N}\right)^{e-u_0}\prod_{j=1}^{\mu}\tau_j^{u_j}\quad \mbox{times}.$$
Expanding the binomial coefficients, we obtain the desired conclusion.
\end{proof}

In theory, when $t=e$ and the Gaussian periods of order $N$ are known, by Theorem \ref{thm-e2} the weight distribution
of the cyclic code $\mathcal{C}$ might be formulated. However, the situation could be quite complicated when $e$ is large or the Gaussian periods have many different values. We list below some special cases in which the weight distribution can be obtained from Theorem \ref{thm-e2}.

If $N=\gcd(\frac{r-1}{q-1},ae)=2$, then $p,q,r$ are all odd and $2|m$. By Lemma \ref{lem-degree2}, the Gaussian periods of order $2$ take on two distinct values $\eta_1=\frac{-1+r^{1/2}}{2},\eta_2=\frac{-1-r^{1/2}}{2}$, each of which corresponds to $\tau_1=\tau_2=1$ cyclotomic class. Hence we have the following corollary.

\begin{cor}\label{thm-eN=2}
When $t=e$ and $N=2$, the cyclic code $\mathcal{C}$ of (\ref{def}) is an $[n,em]$ code over $\mathbb{F}_q$ with at most $\binom{e+2}{2}-1$ nonzero weights. Moreover, the weight distribution of $\mathcal{C}$ is listed in Table \ref{Table-eN=2}.
\end{cor}
\begin{table}[ht]
\caption{The weight distribution of $\mathcal{C}$ when $e=t,N=2$.}\label{Table-eN=2}
\begin{center}{
\begin{tabular}{|c|c|}
  \hline
  Weight & Frequency $(u_0+u_1+u_2=e)$ \\
  \hline
  \hline
  $\frac{(q-1)}{\delta eq}[u_1(r+\sqrt{r})+u_2(r-\sqrt{r})]$ & $\frac{e!}{u_0!u_1!u_2!}\left(\frac{r-1}{2}\right)^{u_1+u_2}$ times \\
  \hline
\end{tabular}}
\end{center}
\end{table}

We remark that Theorem 6 in \cite{Ding1} is a special case of Corollary \ref{thm-eN=2} with $e=t=N=2$.

\begin{exam}
Let $(q, m, e, t)=(7,2,2,2)$. Let $\gamma$ be the generator of $\gf_r^*$ with $\gamma^2 + 6\gamma + 3=0$.
Let $a=1$. Then $N=2$, $(a_1, a_2)=(1,25)$ and
$$
h_{a_1}(x)=x^2 + 2x + 5, \ h_{a_2}(x)= x^2 + 5x + 5.
$$
The parity-check polynomial of $\mathcal{C}$ is then
$
h(x)=x^4 + 6x^2 + 4.
$
The code $\mathcal{C}$ is a $[48,4,18]$ cyclic code over $\gf_7$ with weight enumerator
$
1 + 48z^{18} + 48z^{24}  + 576z^{36} + 1152z^{42} + 576z^{48}.
$
\end{exam}

\vskip3mm
If $N\mid (p^j+1)$ for some positive integer $j$, let $j$ be the least such and let $v=sm/2j$. From Lemma \ref{lem-semip}, the Gaussian periods of order $N$ take on two distinct values $\eta_1=\frac{-1-(-1)^{v}(N-1)r^{1/2}}{N},\eta_2=\frac{-1+(-1)^vr^{1/2}}{2}$, which correspond to $\tau_1=1$ and $\tau_2=N-1$ cyclotomic classes respectively. Hence we have the following corollary.

\begin{cor}\label{thm-eNs}
When $t=e$ and $N\mid (p^j+1)$ for some positive integer $j$, let $j$ be the least such and let $v=sm/2j$, Then the cyclic code $\mathcal{C}$ of (\ref{def}) is an $[n,em]$ code over $\mathbb{F}_q$ with at most $\binom{e+2}{2}-1$ nonzero weights. The weight distribution of $\mathcal{C}$ is listed in Table \ref{Table-eNs}.
\end{cor}
\begin{table}[ht]
\caption{The weight distribution of $\mathcal{C}$ in semiprimitive case and $e=t$.}\label{Table-eNs}
\begin{center}{
\begin{tabular}{|c|c|}
  \hline
  Weight & Frequency $(u_0+u_1+u_2=e)$ \\
  \hline
  \hline
  $\frac{(q-1)}{\delta eq}[u_1(r+(-1)^v(N-1)\sqrt{r})+u_2(r-(-1)^v\sqrt{r})]$ & $\frac{e!}{u_0!u_1!u_2!}\left(\frac{r-1}{2}\right)^{u_1+u_2} (N-1)^{u_2}$ times \\
  \hline
\end{tabular}}
\end{center}
\end{table}

We remark that Theorems 7 and 8 in \cite{Ding2} is a special case of Corollary \ref{thm-eNs} with $e=t=2$.

\begin{exam}
Let $(q, m, e, t)=(5,2,3,3)$. Let $\gamma$ be the generator of $\gf_r^*$ with $\gamma^2 + 4\gamma + 2=0$.
Let $a=1$. Then $N=3$, $(a_1, a_2, a_3)=(1,9,17)$ and
$$
h_{a_1}(x)=x^2 + 2x + 3, \ h_{a_2}(x)= x^2 + 3, \ h_{a_3}(x)= x^2 + 3x+3.
$$
The parity-check polynomial of $\mathcal{C}$ is then
$
h(x)=x^6 + 2.
$
The code $\mathcal{C}$ is a $[24,6,4]$ cyclic code over $\gf_5$ with weight enumerator
$$
1+ 24z^4 + 240z^8+  1280z^{12} + 3840z^{16} + 6144z^{20} + 4096^{24}.
$$
\end{exam}

\vskip3mm
If $N=3$ and $p\equiv 1\pmod{3}$, then $3|m$. By Lemma \ref{lem-period2}, the Gaussian periods of order 3 take
on three distinct values $\eta_1=\frac{-1-c_1r^{1/3}}{3}, \eta_2=\frac{-1+\frac{1}{2}(c_1+9d_1)r^{1/3}}{3}, \eta_3=\frac{-1+\frac{1}{2}(c_1-9d_1)r^{1/3}}{3}$, each of which corresponds to $\tau_1=\tau_2=1$ cyclotomic class, where $c_1$ and $d_1$ are given by Lemma \ref{lem-period2}. Hence we have the following result.

\begin{cor}\label{thm-eN=3}
When $t=e$, $N=3$ and $p\equiv 1\pmod{3}$, the cyclic code $\mathcal{C}$ of (\ref{def}) is an $[n,em]$ code over $\mathbb{F}_q$ with at most $\binom{e+3}{3}-1$ nonzero weights. Moreover, the weight distribution of $\mathcal{C}$ is listed in Table \ref{Table-eN=3}, where $\eta_0,\eta_1,\eta_{2}$ are defined above.
\end{cor}
\begin{table}[ht]
\caption{The weight distribution of $\mathcal{C}$ when $e=t,N=3$.}\label{Table-eN=3}
\begin{center}{
\begin{tabular}{|c|c|}
  \hline
  Weight & Frequency $(u_0+u_1+u_2+u_3=e)$ \\
  \hline
  \hline
  $\frac{(q-1)}{\delta eq}\sum_{j=1}^3u_j(r-1-3\eta_j)$ & $\frac{e!}{u_0!u_1!u_2!u_3}\left(\frac{r-1}{3}\right)^{u_1+u_2+u_3}$ times \\
  \hline
\end{tabular}}
\end{center}
\end{table}

We remark that Theorem 9 in \cite{Ding2} is a special case of Corollary \ref{thm-eN=3} with $e=t=2$ and $N=3$.

\begin{exam}
Let $(q, m, e, t)=(7,3,3,3)$. Let $\gamma$ be the generator of $\gf_r^*$ with $\gamma^3 + 6\gamma^2 + 4=0$.
Let $a=1$. Then $N=3$, $(a_1, a_2, a_3)=(1,115,229)$ and
$$
h_{a_1}(x)=x^3 + 5x + 2, \ h_{a_2}(x)= x^3 +3x + 2, \ h_{a_3}(x)= x^3 + 6x+2.
$$
The parity-check polynomial of $\mathcal{C}$ is then
$
h(x)=x^9+6x^6 + 4x^3+1.
$
The code $\mathcal{C}$ is a $[342,9,90]$ cyclic code over $\gf_7$ with weight enumerator
\begin{eqnarray*}
1+ 342z^{90} + 342z^{96} + 342z^{108} + 38988^{180} + 77976z^{186}+
38988z^{192} + 77976z^{198} + \\
77976z^{204} + 38988z^{216} +
1481544z^{270} +
4444632z^{276} + 4444632z^{282} + 5926176^{288} + \\
8889264z^{294} + 4444632z^{300}
+ 4444632z^{306} + 4444632z^{312} + 1481544z^{324}.
\end{eqnarray*}
\end{exam}

\vskip3mm
If $3\neq N=\gcd(\frac{r-1}{q-1},ae)$ is a prime $\equiv 3\pmod{4}$, $p$ is a quadratic residue modulo $N$
and $\frac{N-1}{2}\mid sm$, let $k=\frac{2sm}{N-1}$, then, according to Lemma \ref{index2-period}, the Gaussian periods take on three values $\eta_1=\eta_0^{(N,r)},\eta_2=\eta_1^{(N,r)},\eta_3=\eta_{-1}^{(N,r)}$, which corresponds to $\tau_1=1$ and $\tau_2=\tau_3=(N-1)/2$ cyclotomic classes respectively. Hence we have the following corollary.

\begin{cor}\label{thm-eindex2}
If $t=e$, $3\neq N\equiv 3\pmod{4}$ is a prime, $p$ is a quadratic residue modulo $N$ and $\frac{N-1}{2}\mid sm$, let $k=\frac{2sm}{N-1}$. Then the cyclic code $\mathcal{C}$ defined in (\ref{def}) is an $[n,em]$ code with at most $\binom{e+3}{3}-1$ nonzero weights, and for each set $\{u_0,u_1,u_2,u_3\}$ of nonnegative integers with  $u_0+u_1+u_2+u_3=e$, the weight distribution of $\mathcal{C}$ is listed in the Table \ref{Table-eindex2}, where $\eta_1,\eta_2,\eta_{3}$ are defined above.
\end{cor}
\begin{table}[ht]
\caption{The weight distribution of $\mathcal{C}$ in the case of index 2 and $e=t$.}\label{Table-eindex2}
\begin{center}{
\begin{tabular}{|c|c|}
  \hline
  Weight & Frequency $(u_0+u_1+u_2+u_3=e)$ \\\hline\hline
  $\frac{(q-1)}{q}[\frac{e-u_0}{\delta e}(r-1)-\frac{N}{\delta e}(u_1\eta_1+u_2\eta_2+u_3\eta_{3})]$ & $\frac{e!}{u_0!u_1!u_2!u_3}\left(\frac{r-1}{N}\right)^{e-u_0}\left(\frac{N-1}{2}\right)^{u_2+u_3}$ times \\\hline
\end{tabular}}
\end{center}
\end{table}

We remark that the main result in \cite{F-M12} is a special case of Corollary \ref{thm-eindex2} with $e=t=2$.

\begin{exam}
Let $(q, m, e, t)=(2,6,7,7)$. Let $\gamma$ be the generator of $\gf_r^*$ with
$\gamma^6 + \gamma^4 + \gamma^3 + \gamma + 1=0$.
Let $a=1$. Then $N=7$ and $p=2$, which is a quadratic residue modulo $N$.
In this case, $(a_1, a_2, a_3, a_4, a_5, a_6, a_7)=(1,10,19,28,37,46,55)$ and
\begin{eqnarray*}
&& h_{a_1}(x)=x^6 + x^5 + x^3 + x^2 + 1, \\
&& h_{a_2}(x)=x^6 + x^5 + 1, \\
&& h_{a_3}(x)=x^6 + x^5 + x^2 + x + 1, \\
&& h_{a_4}(x)=x^6 + x^3 + 1, \\
&& h_{a_5}(x)=x^6 + x^5 + x^4 + x + 1, \\
&& h_{a_6}(x)=x^6 + x + 1, \\
&& h_{a_7}(x)=x^6 + x^4 + x^3 + x + 1.
\end{eqnarray*}
The parity-check polynomial of $\mathcal{C}$ is then
$
h(x)=x^{42} + x^{21} + 1.
$
The code $\mathcal{C}$ is a $[63,42,2]$ cyclic code over $\gf_2$ with weight enumerator
\begin{eqnarray*}
1 + 63z^2 + 1890z^4 + 35910z^6 + 484785z^8 + 4944807z^{10} + 39558456z^{12} +  254304360z^{14} + \\  1335097890z^{16} +  5785424190z^{18} + 20827527084z^{20} + 62482581252z^{22} + \\
156206453130z^{24} + 324428787270z^{26} + 556163635320z^{28} + 778629089448z^{30} + \\
875957725629z^{32} + 772903875555z^{34} + 515269250370z^{36} +
244074908070z^{38} + \\
73222472421z^{40} + 10460353203z^{42}.
\end{eqnarray*}
\end{exam}

\subsection{The case of $2 \leqslant t < e$.}\label{sec-t}

In this section, we consider the case that $2 \leqslant t\leqslant e$. The $t$ zeros of the parity-check polynomial of
$\mathcal{C}$ are $\gamma^{-a_1}, \ldots,\gamma^{-a_t}$,
where $a_j \equiv a+\frac{r-1}{e}\Delta_j \pmod{r-1},\, 1 \leqslant j\leqslant t$. We may assume that $ 0 \leqslant \Delta_1< \Delta_2<\Delta_3<\cdots<\Delta_{t}\leqslant e-1$. Note that each $a_j$ corresponds to the $(\Delta_j+1)$-th column of the matrix $A$ defined in (\ref{matrix}). This is equivalent to choosing an $e\times t$ sub-matrix of $A$, denoted as $B$.
It is possible to choose these $\Delta_i$'s so that any $t$ rows of the matrix $B$ are linear independent over $\gf_q$.
The following lemma demonstrates one way of choosing such $\Delta_i$'s.

\begin{lemma}\label{lem-xyang}
Let $2 \leqslant t \leqslant  e$. Collect any $t$ consecutive columns (modulo $e$) of $A$ defined in (\ref{matrix}) to form matrix $B$. More specifically, for any $\rho$ such that $1\leqslant \rho\leqslant e$, collect the $\bar\rho$-th,$(\overline{\rho+1})$-th,$\cdots$,$(\overline{\rho+t-1})$-th columns of $A$ to form $B$, where $\bar i$ denotes the integer such that $1\leqslant \bar i\leqslant e$ and $\bar i\equiv i \pmod{e}$ for any integer $i$.  Then any $t$ rows of $B$ are $\mathbb{F}_{q}$-linear independent.
\end{lemma}

\begin{proof}
For $0\leqslant i_1<i_2<\cdots <i_t\leqslant e-1$, suppose $B(i_1,\cdots ,i_t)$ is the $(t\times t)$-matrix constituted by the $i_1$-th,$\cdots$,$i_t$-th rows of $B$. Then,
$$\begin{array}{rl}
    B(i_1,\cdots i_t)&=\left(
        \begin{array}{llll}
          \beta^{i_1\bar\rho} & \beta^{i_1(\overline{\rho+1})} & \cdots & \beta^{i_1(\overline{\rho+t-1})}\\
          \beta^{i_2\bar\rho} & \beta^{i_2(\overline{\rho+1})} & \cdots & \beta^{i_2(\overline{\rho+t-1})}\\
          \vdots  & \vdots & &\vdots \\
          \beta^{i_t\bar\rho} & \beta^{i_t(\overline{\rho+1})} & \cdots & \beta^{i_t(\overline{\rho+t-1})}\\
        \end{array}
      \right)\\
&=\left(
   \begin{array}{cccc}
     \beta^{i_1\rho} &  &  &  \\
      & \beta^{i_2\rho} &  &  \\
      &  & \ddots &  \\
      &  &  & \beta^{i_t \rho} \\
   \end{array}
 \right)
 \left(
   \begin{array}{llll}
     1&\beta^{i_1}&\cdots&\beta^{i_1(t-1)}\\
     1&\beta^{i_2}&\cdots&\beta^{i_2(t-1)}\\
      \vdots&\vdots  &  & \vdots \\
     1&\beta^{i_t}&\cdots&\beta^{i_t(t-1)}\\
   \end{array} \right)
\end{array}$$
Since the last matrix in the above formula is a Vandermoned matrix and $1,\beta,\beta^2,\cdots ,\beta^{e-1}$ are nonzero and distinct, $B(i_1,\cdots ,i_t)$ is invertible. This completes the proof of the lemma.
\end{proof}


\vspace{3mm}
\subsubsection{The subcase of $2 \leqslant t \leqslant  e$ and $N=1$}

\begin{thm}\label{thm-tN=1}
Under the Main Assumptions, when $N=1$ and $2 \leqslant t \leqslant e$, and assume that any $t$ rows of the corresponding matrix $B$ are linearly independent. Then the weight distribution of the cyclic code $\mathcal{C}$ defined in (\ref{def}) is listed in Table \ref{Table3}. It is a $t$-weight $[n,tm,d]$ code with $d=\frac{(q-1)r}{\delta eq}(e-t+1)$.
\end{thm}
\begin{table}[ht]
\caption{The weight distribution of $C$ when $N=1$ and $2 \leqslant t \leqslant e$.}\label{Table3}
\begin{center}{
\begin{tabular}{|c|c|}
  \hline
  Weight & Frequency $(1 \leqslant u\leqslant t)$ \\
  \hline
  \hline
  $\frac{(q-1)r}{\delta eq}\cdot (e-t+u)$ & $\binom{e}{t-u}\sum_{k=0}^{u-1}(-1)^k\binom{e-t+u}{k}
  (r^{u-k}-1)$ times \\
  \hline
  $0$ & once \\
  \hline
\end{tabular}}
\end{center}
\end{table}

\begin{proof}
It suffices to compute the value distribution of $T(\vec{x})$ for $\vec{x}=(x_1,\ldots,x_t) \in \mathbb{F}_r^t$.
For any $h$ with $1 \leqslant h \leqslant t$, define
\[L_h:=\left\{\vec{x}=(x_1,\ldots,x_t) \in \mathbb{F}_r^t: \sum_{i=1}^t x_i \beta_i^h=0 \right\},\]
and for any subset $E \subset \{0,1,\ldots,e-1\}$, define
\[\bar{E}:=\{0,1,\ldots,e-1\} \setminus E. \]
We also define
\[N_E:=\bigcap_{h \in E}L_h \setminus \{\vec{0}\}, \quad U_E:=\bigcup_{h \in E}L_h.\]
When $N=1$, the modified Gaussian periods have two possible values $\bar\eta^{(1,r)}_{v}=
\left\{\begin{array}{ll}
 r-1, & \hbox{if $v=0$;} \\
 -1, & \hbox{if $v\in\mathbb{F}_{r}^*$.}
  \end{array}\right.$ Hence by (\ref{equ-Tx}), for any $\vec{x} \in N_E \setminus U_{\bar{E}}$, we have
 \[T(\vec{x})=(\#E)(r-1)+(e-\#E)(-1)=(\#E)r-e. \]
So we only need to compute the order of the set $N_E \setminus U_{\bar{E}}$ for any $E \subset \{0,1,\ldots,e-1\}$ with $\#E$ fixed. Since any $t$ rows of $B$ are linearly independent, and $N_E$ is a vector space over $\mathbb{F}_r$ minus the origin, we have \[N_E=\emptyset \,\, \mbox{ if } \,\, \#E \geqslant t. \]
Now suppose $\#E=t-u$ for some $1 \leqslant u \leqslant t$, then $\#\bar{E}=e-t+u$, and we have $\#N_E=r^{u}-1$. For each $h \in \bar{E}$, for simplicity we define
\[E_h:=N_E \bigcap L_h=N_{E \bigcup \{h\}},\]
then clearly
\[N_E \bigcap U_{\bar{E}}=\bigcup_{h \in \bar{E}} \left(N_E \bigcap L_h\right)=\bigcup_{h \in \bar{E}} E_h.\]
It then follows from the inclusion-exclusion principle that
\[\#\left(N_E \bigcap U_{\bar{E}}\right)=\sum_{k=1}^{u}(-1)^{k+1}\left(
\sum_{\substack{i_1,\ldots,i_k \in \bar{E} \\
\mbox{\tiny distinct} }} \#\left(E_{i_1} \bigcap E_{i_2} \cdots \bigcap E_{i_k} \right)\right).\]
Since
\[\#\left(E_{i_1} \bigcap E_{i_2} \bigcap \cdots \bigcap E_{i_k} \right)=\#\left(N_{E \bigcup \{i_1,\ldots,i_k\}}\right)=r^{u-k}-1,\]
and $\#\bar{E}=e-t+u$, we have
\[\#\left(N_E \bigcap U_{\bar{E}}\right)=\sum_{k=1}^{u}(-1)^{k+1}
\binom{e-t+u}{k}(r^{u-k}-1).\]
We conclude that
\[\#\left(N_E-U_{\bar{E}}\right)=\#N_E-\#\left(N_E \bigcap U_{\bar{E}}\right)=\sum_{k=0}^{u-1}(-1)^{k}
\binom{e-t+u}{k}(r^{u-k}-1).\]
The number of subsets $E \subset \{0,1,\ldots,e-1\}$ such that $\#E=t-u$ is clearly $\binom{e}{t-u}$. This completes the proof of Theorem \ref{thm-tN=1}.
\end{proof}

\begin{rem}
(1). When $e=t$, it is easy to check that
\[\binom{e}{e-u}\sum_{k=0}^{u-1}(-1)^k
\binom{u}{k}(r^{u-k}-1)=\binom{e}{u}(r-1)^u.\]
This is consistent with Theorem \ref{thm-e1}.

(2). When $t=2, e\geqslant 2$, Theorem 5 in \cite{Ding1} is a special case of our Theorem \ref{thm-tN=1}.

(3). Lemma \ref{lem-xyang} justifies the usefulness of Theorem \ref{thm-tN=1}.


\end{rem}

\begin{exam}
Let $(q, m, e, t)=(5,3,4,3)$. Let $\gamma$ be the generator of $\gf_r^*$ with
$\gamma^3 + 3\gamma + 3=0$.
Let $a=1$ and $(\Delta_1, \Delta_2, \Delta_3)=(0, 1, 2)$. Then $N=1$,
$(a_1, a_2, a_3)=(1,32,63)$ and
\begin{eqnarray*}
h_{a_1}(x)=x^3 + x^2 + 2,  \
h_{a_2}(x)=x^3 + 3x^2 + 4, \
h_{a_3}(x)=x^3 + 4x^2 + 3.
\end{eqnarray*}
The parity-check polynomial of $\mathcal{C}$ is then
$
h(x)=x^9 + 3x^8 + 4x^7 + x^6 + x^5 + 4x^4 + x^3 + 2x^2+4.
$
The code $\mathcal{C}$ is a $[124,9,50]$ cyclic code over $\gf_5$ with weight enumerator
\begin{eqnarray*}
1+ 744z^{50} + 61008z^{75} + 1891372z^{100}.
\end{eqnarray*}
\end{exam}

\vskip3mm
\subsubsection{The subcase of $2 \leqslant t \leqslant  e$ and $N\geqslant 2$}

When $t\leqslant e$ and $N\geqslant 2$, the calculation is much more complicated, because there are more Gaussian periods to deal with, so a general result, like Theorem \ref{thm-e2}, could not be obtained. However, some special cases can be treated. Recently, \cite{Tang12} studied the codes in the case of $e=3,t=2,N=2$. They used the theory of elliptic curve. Here using the idea in this paper we give another simple proof, in which we only use results on cyclotomic numbers of order 2.

First, take $\Delta_1=0,\Delta_2=1$. The assumption $2=N=\gcd(\frac{r-1}{q-1},3a)$ implies that $p,q,r$ are odd and $2|a,2|m,2|\delta=\gcd(r-1,a,a+\frac{r-1}{3})$. Then $\beta=\gamma^{\frac{r-1}{3}},g=\gamma^{a},-1=\gamma^{\frac{r-1}{2}}$ all belong to $C_0^{(2,r)}$. Using the relation $$\begin{pmatrix}1&1\\1&\beta\\1&\beta^2\end{pmatrix}\begin{pmatrix}x_1\\x_2\end{pmatrix}
=\begin{pmatrix}y_0\\y_1\\y_2\end{pmatrix},
$$
we know that as $x_1,x_2$ run over $\mathbb{F}_r$, so do $y_0,y_1$, and $y_2=-\beta(y_0+\beta y_1)$. So, we just need to compute the value distribution of
\[ \tilde{T}(y_0,y_1,-\beta y_0-\beta^2 y_1)=\bar{\eta}_{y_0}^{(2,r)}+\bar{\eta}_{y_1}^{(2,r)}
+\bar{\eta}_{ y_0+\beta y_1}^{(2,r)}, \quad (y_0, \ y_1 \in \mathbb{F}_r). \]

If any two of $y_0,y_1,y_0+\beta y_1$ equal to 0, then all of them equal to 0.

If exact one of $y_0,y_1,y_0+\beta y_1$ equals to 0, then we have the following three situations
$$\begin{pmatrix}0\\y_1\\-\beta^2y_1\end{pmatrix},\begin{pmatrix}y_0\\0\\-\beta y_0\end{pmatrix} \mbox{~or~} \begin{pmatrix}-\beta y_1\\y_1\\0\end{pmatrix}.$$
So in this case $\tilde{T}(y_0,y_1,y_0+\beta y_1)$ has two possible values $\bar\eta_0+2\eta_0$ or $\bar\eta_0+2\eta_1$, each of which has frequency $3(r-1)/2$.

If none of $y_0,y_1,y_0+\beta y_1$ equals to 0. Substituting $\beta y_1/y_0$ with $y_1'$, we have
$$\begin{pmatrix}y_0\\y_1\\y_0+\beta y_1\end{pmatrix}=y_0\begin{pmatrix}1\\
\beta^{-1} y'_1\\1+y'_1\end{pmatrix}.$$
Since $y_1'$ and $\beta^{-1} y_1'$ belong to the same $C_i^{(2,r)}$, we have the values and frequencies of $\tilde{T}(y_0,y_1,y_0+\beta y_1)$ below, where the subscript $i=0,1$ are operated modulo 2.
$$    \begin{array}{lll}
Value& Conditions& Frequency\\[3mm]
      3\eta_0, & \hbox{when~}y_0\in C_0^{(2,r)},y'_1\in C_0^{(2,r)},1+y'_1\in C_0^{(2,r)};&\frac{(r-1)}{2}(0,0)^{(2,r)}\mbox{~times};\\[3mm]
      3\eta_1, & \hbox{when~}y_0\in C_1^{(2,r)},y'_1\in C_0^{(2,r)},1+y'_1\in C_0^{(2,r)};&\frac{(r-1)}{2}(0,0)^{(2,r)}\mbox{~times};\\[3mm]
      2\eta_0+\eta_1, & \hbox{when~}y_0\in C_0^{(2,r)},y'_1\in C_i^{(2,r)},1+y'_1\in C_{i+1}^{(2,r)},&\frac{(r-1)}{2}[(0,1)^{(2,r)}+(1,0)^{(2,r)}+(1,1)^{(2,r)}]\mbox{~times};\\
      & \hbox{\qquad or~}y_0\in C_1^{(2,r)},y'_1\in C_1^{(2,r)},1+y'_1\in C_1^{(2,r)};\\[3mm]
      \eta_0+2\eta_1, & \hbox{when~}y_0\in C_1^{(2,r)},y'_1\in C_i^{(2,r)},1+y'_1\in C_{i+1}^{(2,r)},&\frac{(r-1)}{2}[(0,1)^{(2,r)}+(1,0)^{(2,r)}+(1,1)^{(2,r)}]\mbox{~times};\\
      & \hbox{\qquad or~}y_0\in C_0^{(2,r)},y'_1\in C_1^{(2,r)},1+y'_1\in C_1^{(2,r)}.&\\
    \end{array}$$
Then by Lemma \ref{lem-cycNo-N=2}, we have the conclusion below.

\begin{thm}\label{thm-e3t2N2}
If  $e=3,t=2,N=2$, then the cyclic code $\mathcal{C}$ defined in (\ref{def}) is an $[n,2m,\frac{2(q-1)}{3 \delta q}(r-\sqrt{r})]$ code over $\mathbb{F}_q$ with 6 nonzero weights. The weight distribution of $\mathcal{C}$ is listed in Table \ref{Table-e3t2N2}.
\end{thm}
\begin{table}[ht]
\caption{The weight distribution of $\mathcal{C}$ when $e=3,t=2,N=2$.}\label{Table-e3t2N2}
\begin{center}{
\begin{tabular}{|c|c|}
  \hline
  Weight & Frequency \\\hline\hline
  0&once\\\hline
  $\frac{2(q-1)}{3q \delta}(r-\sqrt{r})$ & $\frac{3}{2}(r-1)$ times \\\hline
  $\frac{2(q-1)}{3q \delta}(r+\sqrt{r})$ & $\frac{3}{2}(r-1)$ times \\\hline
  $\frac{(q-1)}{q \delta}(r-\sqrt{r})$ & $\frac{1}{8}(r-1)(r-5)$ times \\\hline
  $\frac{(q-1)}{q \delta}(r+\sqrt{r})$ & $\frac{1}{8}(r-1)(r-5)$ times \\\hline
  $\frac{(q-1)}{q \delta}(3r-\sqrt{r})$ & $\frac{3}{8}(r-1)^2$ times \\\hline
  $\frac{(q-1)}{q \delta}(3r+\sqrt{r})$ & $\frac{3}{8}(r-1)^2$ times \\\hline
\end{tabular}}
\end{center}
\end{table}

\begin{exam}
Let $(q, m, e, t)=(7,2,3,2)$. Let $\gamma$ be the generator of $\gf_r^*$ with
$\gamma^2 + 6\gamma + 3=0$.
Let $a=2$ and $(\Delta_1, \Delta_2)=(0, 1)$. Then $N=2$,
$(a_1, a_2)=(2,18),\ \delta=2,\ n=24$ and
\begin{eqnarray*}
h_{a_1}(x)=x^2 + 6x + 4,  \
h_{a_2}(x)=x^2 + 3x + 1.
\end{eqnarray*}
The parity-check polynomial of $\mathcal{C}$ is then
$
h(x)=x^4 + 2x^3 + 2x^2 + 4x + 4.
$
The code $\mathcal{C}$ is a $[24,4,12]$ cyclic code over $\gf_7$ with weight enumerator
\begin{eqnarray*}
1+ 72z^{12} + 72z^{16} + 264z^{18} + 864z^{20} + 864z^{22}  + 264z^{24}.
\end{eqnarray*}
\end{exam}

Note that for the case of $t=e-1\geqslant 3$, we have found a general method to count the frequency of
$\tilde{T}(y_0,\cdots ,y_{e-1})$. However, there are too many cases to consider and a lot of computation
is involved. We leave this case for future study.

\section{Conclusions}\label{sec-conclusion}

In this paper, we presented a class of cyclic codes $\mathcal{C}$ with arbitrary number of zeros. This
construction is an extension of earlier constructions (see for examples \cite{Ding1,Ding2,Vega12}). In addition,
we determined the weight distribution of $\mathcal{C}$ under the Main Assumptions for
the following special cases:
\begin{itemize}
\item  $t=e$ and the Gaussian periods of order $N$ are known, including the cases that $N=1,2,3$,
           semiprimitive case and a special index 2 case.
\item  $t \leqslant e$, $N=\gcd(\frac{r-1}{q-1},ae)=1$ and any $t$ rows of the matrix $B$ are linearly independent over
           $\mathbb{F}_{q}$.
\item  $t=2,e=3$ and $N=2$ (in this case, we gave a different and simple proof from the main result in \cite{Tang12}).
\end{itemize}

The weight distribution of the code $\mathcal{C}$ is still open in most cases when $t < e$. It would be
good if some of these open cases can be settled.

\subsection*{Acknowledgments}

Cunsheng Ding's and Maosheng Xiong's research are supported by the Hong Kong Research Grants Council under Grant Nos. 600812 and 606211, respectively.
Jing Yang's research  is partly supported by the National Natural Science Foundation of China (No. 10990011, 11001145, 61170289) and the Science and Technology on Information Assurance Laboratory Foundation (No. KJ-12-01).

\end{document}